\documentclass[conference,10pt]{IEEEtran}

\IEEEoverridecommandlockouts                              

\usepackage{framed,color}
\usepackage{verbatim}
\usepackage{graphics} 
\usepackage{epsfig} 
\usepackage{times} 
\usepackage{amsmath} 
\usepackage{amssymb}  
\newtheorem{thm}{Theorem}[section]
\newtheorem{definition}[thm]{Definition}

\newtheorem{lemma}[thm]{Lemma}

\newtheorem{ex}[thm]{Example}
\newtheorem{interpretation}[thm]{Interpretation}
\newtheorem{rem}[thm]{Remark}
\newenvironment{proof}{\begin{IEEEproof}}{\end{IEEEproof}}

\usepackage{epstopdf}

\newcommand{\R}{\mathbb{R}}
\newcommand{\C}{\mathbb{C}}
\newcommand{\Z}{\mathbb{Z}}

\newcommand{\diag}{\text{diag}}

\newcommand{\vol}{\text{Vol}}

\begin{document}

\title{A Comparison of Skewed and Orthogonal Lattices in Gaussian Wiretap Channels}

\author{
\IEEEauthorblockN{Alex Karrila and Camilla Hollanti, \emph{Member, IEEE}} 
\IEEEauthorblockA{Department of Mathematics and Systems Analysis\\ Aalto University, Finland\\ 
        Emails: firstname.lastname@aalto.fi}}
\maketitle
\thispagestyle{empty}
\pagestyle{empty}

\begin{abstract}

We consider lattice coset-coded transmissions over a wiretap channel with additive white Gaussian noise (AWGN). Examining a function that can be interpreted as either the legitimate receiver's error probability or the eavesdropper's correct decision probability, we rigorously show  that, albeit offering simple bit labeling, orthogonal nested lattices are suboptimal for coset coding in terms of both the legitimate receiver's and the eavesdropper's probabilities.

\end{abstract}

\section{Introduction}

We consider a wiretap set-up, in which a message is transmitted to its legitimate receiver Bob in the presence of Eve the eavesdropper. Eve is assumed to have unlimited computational power, but to experience an additional noise compared to Bob. Lattice coset coding is utilized to maximize Eve's confusion, cf. \cite{Wyner,Oggier-Sole-Belfiore}. Bob's lattice is referred to as the code lattice or dense lattice, and Eve's lattice as the sparse or coarse lattice. The channel is assumed to exhibit additive white Gaussian noise (AWGN) but no fading. 
 The respective channel equations for Bob and Eve are

$$
\mathbf{y}_b=\mathbf{x}+\mathbf{n}_b,\quad \mathbf{y}_e=\mathbf{x}+\mathbf{n}_e,
$$
where  $\mathbf{y}$ is the received vector, $\mathbf{x}$ the transmitted coset-coded vector, and $\mathbf{n}$ is AWGN with respective variances $\sigma_b^2 < \sigma_e^2$.

In finding optimal lattice wiretap codes, there are three main objectives:
\begin{itemize}
\item[i)] Maximizing the data rate $R$, which is determined by the size of the codebook $\mathcal{C}$ and the decoding delay $n$ as 
$$R=\frac{\log_2|\mathcal{C}|}n$$
bits per channel use (bpcu).
\item[ii)] Minimizing the legitimate receiver's decoding error probability.
\item[iii)] Minimizing the eavesdropper's probability of correct decision.
\end{itemize}

Considering only the first two problems, the largest codebooks for a fixed transmission power and an upper bound for the receiver's error probability are the solutions to the widely investigated sphere-packing problem. This results in lattices that are typically nonorthogonal (see, e.g., \cite{Conway-Sloane}). Orthogonal lattices have still traditionally been preferred due to an easy-to-implement bit-labeling algorithm, namely the Gray-mapping. Due to this mapping, the encoding and decoding procedures are more straightforward for orthogonal lattices than for nonorthogonal, i.e., skewed lattices. Nevertheless, computationally efficient closest-point algorithms such as the sphere decoder also exist for nonorthogonal lattices (for an explicit construction, see \cite{Viterbo}, Sec. 4). In \cite{CamiKalle,KumarCaire}, it was also demonstrated how skewed lattices can be efficiently encoded and decoded by using a modified power-controlled sphere decoder or sphere decoding adjoined with minimum-mean-square-error generalized-decision-feedback-equalization (MMSE-GDFE), both resulting in optimal (maximum-likelihood) performance. Hence, skewed lattices should not be excluded when searching for optimal lattices, in particular in the light of the present paper showing that they are not only better in terms of Bob's performance, but also in terms of confusing the eavesdropper.

Similarly to the sphere-packing problem in Bob's case, we now include the third objective in our consideration. Our approach is to fix the data rate and the transmission power and then compare skewed and orthogonal lattices from the point of view of the latter two objectives in an AWGN channel. We study an expression that has two alternative interpretations as either the receiver's error probability (REP) for any lattice code in an AWGN channel or the eavesdropper's correct decision probability (ECDP) for a lattice coset code in an AWGN channel. We prove the following results (notation will be defined in the subsequent section).
\begin{itemize}
\item[i)] Skewing Bob's orthogonal code lattice $\Lambda_b$ will decrease the REP of any code.
\item[ii)] Skewing Eve's orthogonal sparse lattice $\Lambda_e$ will decrease the ECDP of any lattice coset code.
\item[iii)] Combining the previous two results, the common set-up of the dense lattice $\Lambda_b$ being orthogonal and the commonly used choice of an orthogonal sublattice $\Lambda_e = 2^k \Lambda_b$ are suboptimal in terms of both the ECDP and the REP. According to whether  Gray-labeling is insisted or not, this common set-up can be improved by either choosing a skewed sublattice of the same orthogonal dense lattice $\Lambda_b$, leaving Bob's lattice orthogonal and the REP suboptimal, or skewing both lattices. 
\end{itemize}
These results suggest that skewed lattices deserve more attention in the study of the AWGN wiretap channels even though their encoding and decoding are admittedly somewhat more complicated than that of orthogonal lattices. It is also worthwhile to keep in mind that in any practical system, an outer error correcting code, e.g., a low-density parity-check (LDPC) code, is used in addition to the inner lattice code. The true decoding bottle-neck in this case is the outer code requiring soft input, not the lattice code. 

\section{Preliminaries}

In this section, we present some necessary definitions and their information-theoretic interpretations.

\begin{definition}
A \textit{lattice} is a discrete additive subgroup of $\R^n$.
\end{definition}

Any point in a lattice $\Lambda \subset \R^n$ can be expressed in terms of a \textit{generator matrix} $M \in \R^{n \times m}$ as follows
\begin{equation*}
\Lambda = \{ \mathbf{x} \in \R^n \vert \mathbf{x} = M \omega, \omega \in \Z^m \}.
\end{equation*}
We assume that the columns of $M$ are linearly independent over $\Z$ and hence, the \textit{lattice coordinates} $\omega$ of a lattice point are unique. If $m=n$, the lattice is of \textit{full rank}. A \textit{sublattice} of a lattice of dimension $m$ in $\R^n$ is an additive subgroup; it has a generator matrix $MZ$, where $Z \in \Z^{m \times k}$. Here $k$ is the dimension of the sublattice and for a square matrix $Z$, $$\vert \Lambda_b / \Lambda_e \vert = \vert \det Z \vert.$$ The \textit{volume} $\vol(\Lambda)$ of the lattice $\Lambda$ is the volume of the fundamental parallellotope spanned by the column vectors of $M$, given by $$\vol(\Lambda)=\vert \det M \vert$$ for full-rank lattices.

\begin{rem}
Differing from some  information theory references, here vectors are identified with \textit{column} matrices and the lattice generator vectors with the \textit{columns} of the generator matrix $M$.
\end{rem}

\begin{definition}
The \textit{dual lattice} $\Lambda^\star$ of a full-rank lattice $\Lambda$ generated by $M$ is the one generated by $$ M^{-T}:=(M^{-1})^T = (M^T)^{-1} .$$
\end{definition}


\begin{thm}[The Poisson formula for lattices] Let $\Lambda$ be a full-rank lattice with generator $M$  and let $f: \R^n \to \C$ be a continuous function with $\int_{\mathbf{x} \in \R^n} \vert f(\mathbf{x}) \vert d^n x < \infty$ and $\sum_{\mathbf{t} \in \Lambda^\star} \vert \hat{f}(\mathbf{t}) \vert < \infty$ such that the partial sums of $\sum_{\mathbf{t} \in \Lambda} \vert f(\mathbf{t} + \mathbf{u}) \vert$ converge uniformly whenever $\mathbf{u}$ is restricted onto a compact set. Then,
\begin{equation*}
\sum_{\mathbf{t} \in \Lambda} f(\mathbf{t}) = \vert \det M \vert ^{-1} \sum_{\mathbf{t}\in \Lambda^\star} \hat{f}(\mathbf{t})
\end{equation*}
where the Fourier transform is defined as
\begin{equation*}
\hat{f}(\mathbf{t}) = \int_{\mathbf{y} \in \R^n} e^{-i 2 \pi \mathbf{y} \cdot \mathbf{t}} f(\mathbf{y}) dy.
\end{equation*}
\end{thm}

\begin{proof}
The proof is given in \cite{Ebeling}. We point out that the condition on the continuity of $f$ is essential for the proof and is missing in the book.
\end{proof}

The function that we will optimize is the following.

\begin{definition}
The \textit{psi function} $\psi_\Lambda (x)$ of a lattice $\Lambda$ at a point $x \in \R_+$ is given by
\begin{equation*}
\psi_\Lambda(x) =  \sum_{\mathbf{t} \in \Lambda} e^{-x \Vert \mathbf{t} \Vert^2}.
\end{equation*}
\end{definition}
This is a variant of lattice theta series restricted on the imaginary axis, $\psi_\Lambda (x) = \Theta_\Lambda (ix/\pi)$. The convergence properties of the psi series follow from those of the theta series.

\begin{interpretation}
In \cite{Oggier-Sole-Belfiore}, an upper approximation for the ECDP $P_{c,e}$ for a lattice coset code is derived as
\begin{equation}
\label{Pce}
P_{c,e} \le \frac{\vol (\Lambda_b)}{(\sqrt{2 \pi} \sigma_e)^n} \psi_{\Lambda_e} \left( \frac{1}{2 \sigma_e^2} \right).
\end{equation}
Here $\Lambda_b$ is the dense and $\Lambda_e$ the sparse lattice, inteded for the receiver and the eavesdropper, respectively. The lattices are assumed to be of full rank and the eavesdropper's noise is assumed to be AWGN with variance $\sigma_e^2$. The inequality \eqref{Pce} is tight for large $\sigma_e$. For small $\sigma_e$, the upper bound is larger than $1$ and hence useless.

On the other hand, using the union bound technique as is done in \cite[Appendix II]{Boutros-Viterbo-Rastello-Belfiore} for Rayleigh-fading channels and setting the Rayleigh fading coefficients equal to one, the REP can be approximated from above as:
\begin{equation}
\label{Peb}
P_{e, b} \le 1/2 \left( \psi_{\Lambda_b}\left( \frac{1}{8 \sigma_b^2} \right) - 1 \right).
\end{equation}
This formula is valid for any lattice  code $\Lambda_b$ (not just a coset code) in an AWGN channel and the approximation is good for small receiver's noise variances $\sigma_b^2$.

Based on these two formulae and the fact that the variances $\sigma_e^2$ and $\sigma_b^2$ vary with the random channels, our subsequent aim will be to provide inequalities of the form $\psi_{\Lambda_1}(x) < \psi_{\Lambda_2}(x)$ for all $x \in \R_+$. When comparing different lattices sharing the same dimension, their volumes are first normalized to one. This ensures that for a relatively large fixed transmission power, the finite codebooks carved from the infinite lattices will be approximately equally large, and hence we can fairly compare the lattice codes without considering the actual data rates, as these will coincide.
\end{interpretation}

\begin{rem}
Due to the obvious connection between the formulae \eqref{Peb} and \eqref{Pce} for the REP and ECDP, respectively, one would intuitively guess that a solution for the sphere-packing problem also yields an optimal ECDP. This, however, does not seem to work on the level of mathematical proofs; Eq. \eqref{Peb} is obtained by the union bound technique, whereas in the sphere-packing problem, the upper bound for REP is based on integrating a Gaussian function over a ball, yielding a much tighter bound for large receiver's noise variances $\sigma_b$ or, equivalently, for small arguments of $\psi$. To minimize the ECDP, we want to minimize $\psi$ for small arguments. Hence, even if the sphere-packing probability bound is small, it does not provide us with immediate information as to how small the $\psi$ function is for small arguments, i.e., how small the ECDP is.
\end{rem}

%
%

\section{Skewing an orthogonal lattice}

In this section, we show that skewing a lattice will always improve a code both in terms of Eve's and Bob's probabilities.

\begin{lemma}
\label{translation lemma}
For any full-rank lattice $\Lambda$,
\begin{equation}
\psi_\Lambda (x) > \sum_{\mathbf{t} \in \Lambda} e^{-x \Vert \mathbf{t} + \mathbf{u} \Vert^2}
\end{equation}
for any $\mathbf{u} \not \in \Lambda$.
\end{lemma}

\begin{proof}
Denote summands of the respective sides as $g(\mathbf{t}) = e^{-x \Vert \mathbf{t} \Vert^2}$ and $f(\mathbf{t}) = e^{-x \Vert \mathbf{t} + \mathbf{u} \Vert^2}$, so $f(\mathbf{t}) = g(\mathbf{t} + \mathbf{u})$. Then, by the elementary properties of Fourier transform, we have $\hat{f}(\mathbf{t}) = \hat{g}(\mathbf{t}) e^{-i 2 \pi \mathbf{t} \cdot \mathbf{u}}$. Hence, using the Poisson formula,
\begin{eqnarray*}
&& \sum_{\mathbf{t} \in \Lambda} e^{-x \Vert \mathbf{t} + \mathbf{u} \Vert^2}\\
 &=& \sum_{\mathbf{t} \in \Lambda} f(\mathbf{t}) \\
&=& \vert \det M \vert^{-1} \sum_{\mathbf{t} \in \Lambda^\star} \hat{g}(\mathbf{t}) e^{-i 2 \pi \mathbf{t} \cdot \mathbf{u}},
\end{eqnarray*}
where $M$ is the generator matrix of $\Lambda$. In continuation, we will use the knowledge that the Fourier transform of the gaussian function $g$ is another gaussian, hence a real, positive and even function. (The explicit form of $\hat{g}$ could be calculated but it is not necessary.) First, since $\hat{g}$ is even, the imaginary parts $-i \hat{g}(\pm \mathbf{t}) \sin (\pm 2 \pi \mathbf{t} \cdot \mathbf{u})$ of the summand for lattice $\Lambda^\star$ points $\pm \mathbf{t}$ cancel out, yielding
\begin{eqnarray*}
&& \sum_{\mathbf{t} \in \Lambda} e^{-x \Vert \mathbf{t} + \mathbf{u} \Vert^2}\\
&=& \vert \det M \vert^{-1} \sum_{\mathbf{t} \in \Lambda^\star} \hat{g}(\mathbf{t}) \cos(2 \pi \mathbf{t} \cdot \mathbf{u}). \end{eqnarray*}

Next, we need the positivity of $\hat{g}$ to be able to approximate the cosine by $1$. First, note that we assumed $\mathbf{u} \not \in \Lambda$, equivalently, $\mathbf{u} = M \omega_1$ with some component of $\omega_1$, say the $j^{th}$ one $\omega_{1,j}$, not integer. Also note that $\Lambda^\star \ni \mathbf{t} = M^{-T}\omega_2,$ where $ \omega_2 \in \Z^n$. Hence, $\mathbf{t} \cdot \mathbf{u} = \omega_2^T M^{-1} M \omega_1 = \omega_2^T \omega_1$. Now, choosing the lattice point $\mathbf{t}$ such that $\omega_2=\mathbf{e}_j$, we immediately see that $\mathbf{t} \cdot \mathbf{u} = \omega_{1,j} \not \in \Z$ and $\cos(2 \pi \mathbf{t} \cdot \mathbf{u}) < 1$. Hence, replacing $\cos(2 \pi \mathbf{t} \cdot \mathbf{u})$ by $1$ in the preceding step, we get a strict inequality
\begin{eqnarray}
&& \sum_{\mathbf{t} \in \Lambda} e^{-x \Vert \mathbf{t} + \mathbf{u} \Vert^2}\\
& < & \vert \det M \vert^{-1} \sum_{\mathbf{t} \in \Lambda^\star} \hat{g}(\mathbf{t}) \label{poisson} \\
 &=& \vert \det M \vert ^{-1} \vert \det M^{-T} \vert^{-1} \sum_{\mathbf{t} \in \Lambda} \hat{\hat{g}}(\mathbf{t}) \\
&=& \sum_{\mathbf{t} \in \Lambda} \hat{\hat{g}}(\mathbf{t}),
\end{eqnarray}
where we have again applied the the Poisson formula to \eqref{poisson}. Finally, the double Fourier transform is in general a reflection operator, so $\hat{\hat{g}}(\mathbf{t}) = g(-\mathbf{t})$, and using the fact that $g(-\mathbf{t}) = g(\mathbf{t})$ we obtain the result,
\begin{eqnarray*}
&& \sum_{\mathbf{t} \in \Lambda} e^{-x \Vert \mathbf{t} + \mathbf{u} \Vert^2}\\
& < & \sum_{\mathbf{t} \in \Lambda} g(\mathbf{t}) \\
& = & \psi_\Lambda (x).
\end{eqnarray*}
\end{proof}

\begin{definition}
\label{def:skew}
Let $\Lambda_o$ be a full-rank orthogonal lattice in $\R^n$ with generator vectors $a_1 \mathbf{e}_1,...,a_n\mathbf{e}_n$, $a_i > 0$ for $1 \le i \le n$. We call a lattice $\Lambda_s \neq \Lambda_o$ a \textit{skewing} of $\Lambda_o$, if it has a generator matrix that is an upper triangular matrix with the diagonal elements $a_1,...,a_n$.
\end{definition}

This definition has a simple geometric interpretation, depicted in Fig. \ref{skewfig}.

\begin{figure}[h!]
  
  \centering
    \includegraphics[width=0.5\textwidth]{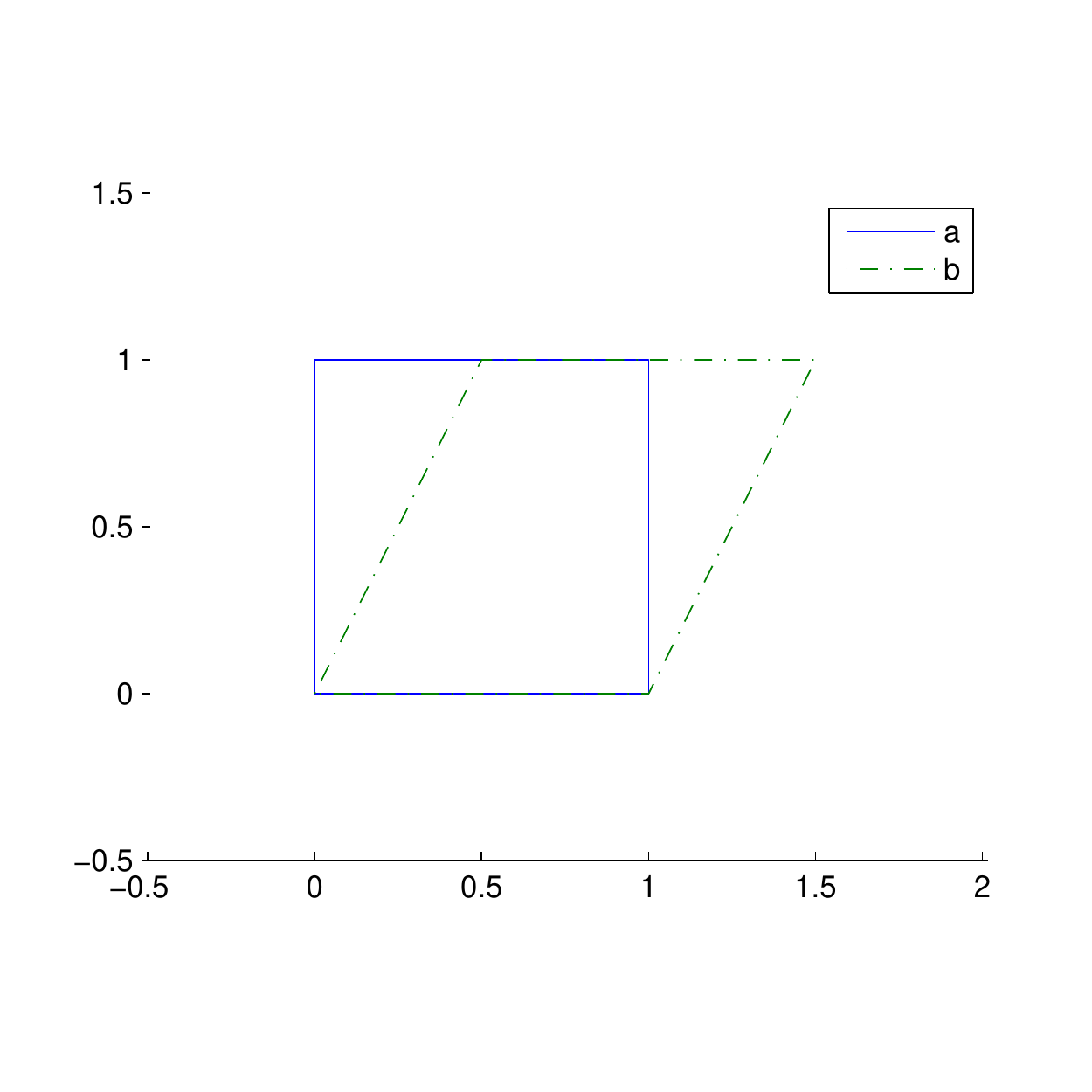}
    \caption{The fundamental parellellotopes of (a) a square lattice (b) its skewing.}
    \label{skewfig}
\end{figure}

We point out that skewing can be interpreted as a matrix operation. If $M_o$ and $M_s$ are the generator matrices of $\Lambda_o$ and $\Lambda_s$, respectively, then the non-singular skewing matrix $S$ can be solved from the matrix equation
\begin{equation}
M_s = S M_o.
\end{equation}
This equation is non-singular, since $$\det M_o = \prod_{i=1}^n a_i \neq 0$$ and $$ \det M_s = \prod_{i=1}^n a_i = \det M_o$$ by the determinant rule of upper triangular matrices. This also implies that $\det S = 1$.

Now we are ready to state the main theorem. After this, we will provide an illustrative interpretation of the theorem and prove it.

\begin{thm}
\label{skewing theorem}
For a skewing $\Lambda_s$ of a full-rank orthogonal lattice $\Lambda_o$,
\begin{equation*}
\psi_{\Lambda_s} (x) < \psi_{\Lambda_o} (x)
\end{equation*}
for all $x > 0$.
\end{thm}

\begin{interpretation}
Skewings provide several easy ways to improve lattice coset codes. We point out that since skewing keeps the lattice volume constant ($\det S = 1$ in matrix representation), it will not affect the size of a spherical codebook. Hence, a lattice comparison between skewings only requires considering the ECDP and the REP. With ths knowledge, the theorem has the following immediate implications.
\begin{itemize}
\item[i)] Comparing a dense lattice $\Lambda_{b, o}$ and its skewings $\Lambda_{b, s}$, Theorem \ref{skewing theorem} applied to Eq. \eqref{Peb} shows that the REP is always smaller for the skewings $\Lambda_{b, s}$. This holds for all codes, not just coset codes.

\item[ii)] Consider a coset code arising from a fixed nonorthogonal lattice $\Lambda_b$. Then, to minimize the ECDP \eqref{Pce}, it seems that $\Lambda_e$ should not be chosen orthogonal (if orthogonal sublattices exist). Note that then no skewing of the orthogonal $\Lambda_e$ is necessarily a sublattice of $\Lambda_b$, so this is just heuristics.

\item[iii)] Consider a typical set-up of $\Lambda_{b, o}$ generated by $M = \diag(a_1, ... a_n)$ being orthogonal and $\Lambda_{e, o} = 2^k \Lambda_{b, o}$. In this case both the REP and the OCDP are suboptimal. There are two remedies:
\begin{itemize}
\item First, we can skew both $\Lambda_{e, o}$ and $\Lambda_{b, o}$. Skewing by $S$ so that the skewed lattices $\Lambda_{b,s}$ and $\Lambda_{e,s}$ are generated by $SM$ and $2^k SM$, respectively, will yield a nonorthogonal lattices but preserve the volumes: $\vol(\Lambda_{b, s}) = \vol (\Lambda_{b, o})$ (since $\det S = 1$). Hence, applying this and Theorem \ref{skewing theorem} in Eqs. \eqref{Peb} and \eqref{Pce}, we see that skewing will decrease both the REP and the ECDP. However, the skewed lattice will not allow for a simple Gray mapping, or in other words, the Gray mapping is not guaranteed to give an optimal bit-labeling.
\item
Second, we can only opt for skewing the sublattice $\Lambda_e$, while leaving $\Lambda_b$ orthogonal. This means that the REP will remain suboptimal, but the lattice will allow for Gray labeling and maintains simpler encoding and decoding for Bob along the lines discussed in the introduction. Moreover, the ECDP is decreased. The idea is that if $\Lambda_{b, o}$ is orthogonal and generated by $\diag(a_1, ... a_n)$, and $\Lambda_{e, o} = 2^k \Lambda_b$, then any  sublattice $\Lambda_{e,s}$ generated by $MZ$, where $Z$ is an upper triangular integer matrix with diagonal etries $2^k$, is easily proven to be a skewing of $\Lambda_{e,o}$ (or equal to $\Lambda_{e,o}$). Then, applying Theorem \ref{skewing theorem} to Eq. \eqref{Pce}, we see that $\Lambda_{e, s}$ will yield a lower ECDP.
\end{itemize}
\end{itemize}

\end{interpretation}

\begin{proof}[Proof of Theorem \ref{skewing theorem}]

Let us use the notation of Definition \ref{def:skew}. Furthermore, denote by $\Lambda_{s,k}$ the embedding into $\R^k$, $k \le n$, of $\Lambda_s \cap \R^k \times 0^{n-k}$. Equivalently, $\Lambda_{s,k}$ is the lattice in $\R^k$ generated by the $k$ first columns of the generator $M_s$ of $\Lambda_s$. Continuing to ease the notation, denote the projection of the $k^{th}$ column of $M_s$ onto $\R^{k-1}$ by $\mathbf{m}_k$, so the $k^{th}$ column is $a_k \mathbf{e}_k + \mathbf{m}_k$.

Now, it is apparent from the definition of $\Lambda_{s, k}$ that
\begin{equation}
\label{psin}
\psi_{\Lambda_{s}} (x) = \psi_{\Lambda_{s, n}} (x),
\end{equation}
and that
\begin{equation}
\label{psi1}
\psi_{\Lambda_{s, 1}} (x)= \sum_{\omega_1 \in \Z}  e^{-x \omega_1^2 a_1^2}.
\end{equation}
On the other hand, using Lemma \ref{translation lemma} for the lattice $\Lambda_{s, k-1}$, $2 \le k \le n$, which is a full-rank lattice of $\R^{k-1}$, we obtain

\begin{eqnarray}
\nonumber
&& \psi_{\Lambda_{s, k}} (x) \\
\nonumber
&=& \sum_{\mathbf{t} \in \Lambda_{s, k}} e^{-x \Vert \mathbf{t} \Vert^2} \\
\nonumber
&=& \sum_{\omega_k \in \Z} \quad \sum_{\mathbf{t}^{(k-1)} \in \Lambda_{s, k-1}} e^{-x \omega_k^2 a_k^2 -x \Vert \mathbf{t}^{(k-1)} + \omega_k \mathbf{m}_k \Vert^2} \\
\nonumber
&=& \sum_{\omega_k \in \Z} e^{-x \omega_k^2 a_k^2} \sum_{\mathbf{t}^{(k-1)} \in \Lambda_{s, k-1}} e^{-x \Vert \mathbf{t}^{(k-1)} + \omega_k \mathbf{m}_k \Vert^2} \\
\nonumber
& \le & \sum_{\omega_k \in \Z}  e^{-x \omega_k^2 a_k^2} \sum_{\mathbf{t}^{(k-1)} \in \Lambda_{s, k-1}} e^{-x \Vert \mathbf{t}^{(k-1)} \Vert^2}\\
\label{induction}
&=& \left( \sum_{\omega_k \in \Z}  e^{-x \omega_k^2 a_k^2} \right) \psi_{\Lambda_{s, k -1}} (x)
\end{eqnarray}
and, as stated in Lemma \ref{translation lemma}, the equality holds if and only if $\omega_k \mathbf{m}_k \in \Lambda_{s, k-1}$ for all $\omega_k \in \Z$, equivalently, $\mathbf{m}_k \in \Lambda_{s, k-1}$. This is furthermore equivalent to that the $k^{th}$ column $a_k \mathbf{e}_k + \mathbf{m}_k$ of $M_s$ can be replaced by $a_k \mathbf{e}_k$ without changing the lattice $\Lambda_s$.

Next, starting from Eq. \eqref{psin}, using the identity \eqref{induction} inductively, and finally using Eq. \eqref{psi1}, we obtain
\begin{eqnarray*}
\psi_{\Lambda_{s, n}} (x) 
& \le & \prod_{i=1}^n \left( \sum_{\omega_i \in \Z}  e^{-x \omega_i^2 a_i^2} \right) \\
&=& \psi_{\Lambda_{s, o}} (x).
\end{eqnarray*}
The equality holds if and only if it has  been possible to modify,  for all $k$,  the $k^{th}$ column of $M_s$ into $a_k \mathbf{e_k}$ without changing the lattice generated by $M_s$. But this is equivalent to $M_s$ and $(a_1 \mathbf{e}_1,...,a_n \mathbf{e}_n) = M_o$ generating the same orthogonal lattice $\Lambda_o$. This is impossible by the definition of a skewing. Hence, for any skewing $\Lambda_s$ of $\Lambda_o$, we have a strict inequality
\begin{equation*}
\psi_{\Lambda_s} (x) < \psi_{\Lambda_o} (x)
\end{equation*}
for all $x > 0$. This completes the proof.
\end{proof}

\begin{ex}
The Gosset lattice $E_8$ has the generator matrix $M_s$ given by \cite{Conway-Sloane}

\begin{equation*}
\left( \begin{array}{r r r r r r r r}
2 & -1 & 0& 0& 0& 0& 0& 1/2 \\
0& 1& -1& 0& 0& 0& 0& 1/2 \\
0& 0& 1& -1& 0& 0& 0& 1/2 \\
0& 0& 0& 1& -1& 0& 0& 1/2 \\
0& 0& 0& 0& 1& -1& 0& 1/2 \\
0& 0& 0& 0& 0& 1& -1& 1/2 \\
0& 0& 0& 0& 0& 0& 1& 1/2 \\
0& 0& 0& 0& 0& 0& 0& 1/2 \\
\end{array} \right),
\end{equation*}
so it is a skewing of the orthogonal lattice $\Lambda$ generated by $M_o = \diag(2, 1, ..., 1, 1/2)$. The theta series of the Gosset lattice is expressible by the Jacobi theta functions as \cite{Conway-Sloane}
\begin{equation*}
\Theta_{E_8}(z) = 1/2 (\vartheta_2(q)^8 + \vartheta_3(q)^8 + \vartheta_4(q)^8),
\end{equation*}
where $q = e^{i \pi z}$. The theta series of the orthogonal lattice $\Lambda$ is
\begin{equation}
\Theta_{\Lambda}(z) = \prod_{j =1}^n \vartheta_3(q_j),
\end{equation}
where $q_j = e^{i \pi a_j^2 z}$ and $a_j$ is the $j^{th}$ diagonal element of $M_o$. Now, recalling that $\psi_\Lambda (x) = \Theta_\Lambda (ix/\pi)$, we can compare the psi series of these two lattices by evaluating Jacobi theta functions. The plots of the psi functions are depicted in Fig. \ref{psi_skewings}. The figure shows that $\psi_{E_8}(x) < \psi_{\Lambda}(x)$ for all $x$, as predicted by Theorem \ref{skewing theorem}.

In coset coding, this has the following interpretation: $E_8$ and $\Lambda$ are both index $2^8$ subgroups of $\frac{1}{2}\Z^8$. If Bob's lattice is $\frac{1}{2}\Z^8$, then the coset lattices $E_8$ and $\Lambda$ will yield the same code rates, but $E_8$ with a better secrecy.
\end{ex}

\begin{figure}[h!]
  
  \centering
    \includegraphics[width=0.5\textwidth]{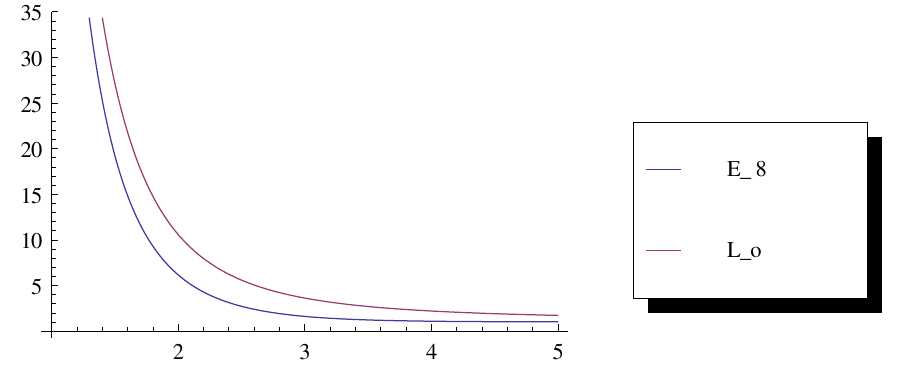}
    \caption{The psi functions of an orthogonal lattice $\Lambda=L_o$ and its skewing $E_8$.}
    \label{psi_skewings}
\end{figure}

\section{Conclusions}

In the construction of lattice codes for AWGN wiretap channels, skewed lattices should be taken more seriously. Namely, we have proved that orthogonal lattices are suboptimal not only in terms of the receiver's error probability as we already know from the sphere-packing theorems, but also in terms of the eavesdropper's correct decision probability when using lattice coset codes. Hence, the design of secure lattice codes should ideally be based on skewed lattices. However, due to implementation purposes, one may opt for only skewing the eavesdropper's lattice, while preserving the orthogonality of the legitimate receiver's lattice, which results in suboptimal performance but easy-to-implement algorithms for Bob, as well as improved security.   

\section{Acknowledgments}

This work was carried out during A. Karrila's MSc thesis project. The Department of Mathematics and Systems Analysis at Aalto University is gratefully acknowledged for the financial support. 

C. Hollanti is financially supported by the Academy of Finland grants \#276031, \#282938, and \#283262, and by a grant
from Magnus Ehrnrooth Foundation, Finland. 

The support from the European Science Foundation under the ESF COST Action IC1104 is
also gratefully acknowledged.

\end{document}